\documentclass[conference]{IEEEtran}

\usepackage[letterpaper, left=0.680in, right=0.680in, bottom=1.049in, top=0.71in]{geometry}

\usepackage{ifpdf}
\ifCLASSINFOpdf
\usepackage[pdftex]{graphicx}
\graphicspath{{./Figures/}}
\DeclareGraphicsExtensions{.pdf,.jpeg,.png}
\else
\usepackage[dvips]{graphicx}
\DeclareGraphicsExtensions{.pdf}
\usepackage[caption=false, font=footnotesize]{subfig}
\fi

\usepackage{amsthm}
\usepackage{epstopdf}
\usepackage{amsmath}
\usepackage{amssymb}
\usepackage{color}
\usepackage{array}
\usepackage{mathtools}
\usepackage[ruled, vlined, linesnumbered]{algorithm2e}
\usepackage{enumitem}
\usepackage[mathcal]{euscript}
\usepackage{tabularx}
\usepackage{multirow}
\usepackage{booktabs}
\usepackage{makecell}
\usepackage{aligned-overset}
\usepackage{bbm}
\usepackage[caption=false,font=footnotesize]{subfig}
\usepackage{url}
\usepackage{tikz}
\usetikzlibrary{shapes.geometric, arrows.meta, positioning, calc}

\newtheorem{lemma}{Lemma}

\newtheorem{corollary}{Corollary}

\newcommand{\Rsym}{\mathcal{R}_{\mathrm{sym}}}

\usepackage{xcolor}

\definecolor{cresppurple}{rgb}{0.5,0,0.5}

\begin{document}

	\title{RIS-Aided ISAC in Cluttered Environments}

	\author{
		\IEEEauthorblockA{
			Yosefine Triwidyastuti\IEEEauthorrefmark{1},
			Tri~Nhu~Do\IEEEauthorrefmark{1},
			Ngo Hoang Tu\IEEEauthorrefmark{2}, and
			Georges~Kaddoum\IEEEauthorrefmark{3}
		}
		\IEEEauthorblockA{
			\IEEEauthorrefmark{1}Telecom Neural Detection (TND) Lab, Dept. of Electrical Engineering, Polytechnique Montr\'{e}al, Montreal, QC, Canada
		}
		\IEEEauthorblockA{
			\IEEEauthorrefmark{2}Faculty of Information Technology, Van Lang School of Technology, Van Lang University, Ho Chi Minh City, Vietnam
		}
		\IEEEauthorblockA{
			\IEEEauthorrefmark{3}Dept. of Electrical Engineering, \'{E}cole de technologie sup\'{e}rieure (\'{E}TS), University of Quebec, Montreal, QC, Canada
		}
		\IEEEauthorblockA{Emails: yosefine-2.triwidyastuti@polymtl.ca, tri-nhu.do@polymtl.ca, tu.nh@vlu.edu.vn, and georges.kaddoum@etsmtl.ca}
	}

	\maketitle

\begin{abstract}
In this paper, we analyze the performance of a communication-optimized reconfigurable intelligent surface (RIS)-assisted integrated sensing and communication (ISAC) system operating in a cluttered environment where multiple scatterers may interfere with the different types of reflected sensing signals. The RIS phases coherently combine the direct and reflected communication paths at the user equipment, whereas the corresponding radar returns remain generally misaligned. In addition, static scatterers near the radar act as environmental clutter that affects only the sensing function. For the communication link over small-scale fading, we derive an exact ergodic-capacity expression for the no-RIS baseline, a moment-matched Gamma approximation for the RIS-assisted link, and a Jensen upper bound, all of which are interpreted as upper bounds on the rate of the underlying binary phase-shift keying waveform. For sensing, our analysis focuses on the average signal-to-clutter-plus-noise ratio (SCNR) at the direct range-Doppler cell. Specifically, we derive the average powers of the direct, RIS-related, and scatterer returns, which scale as constant, linear, linear, quadratic, and constant, respectively, with the number of RIS elements. We then weigh them by the range and slow-time leakage responses to obtain the SCNR, thereby separating RIS-induced clutter from geometry-governed environmental clutter. Range and velocity estimation are evaluated using resolution-normalized metrics. Our Monte Carlo simulation results validate the analysis and show that zero-Doppler clutter leakage dominates the SCNR.
\end{abstract}


\section{Introduction}

By allowing radar sensing and data transmission to share a waveform, spectrum, and hardware platform, integrated sensing and communication (ISAC) promises spectral, energy, and cost savings~\cite{BouzabiaTVT2025}. A further degree of freedom is provided by reconfigurable intelligent surfaces (RISs), whose passive elements impose controllable phase shifts on impinging waves~\cite{DoTCOM2021}. To date, RISs have given rise to the RIS-ISAC paradigm, in which the waveform and surface configuration are jointly designed. However, a surface configuration that favors one function may penalize the other, so this tradeoff must be quantified, rather than assumed.

In this context, the present study investigates a monostatic RIS-assisted ISAC system based on a phase-modulated continuous-wave (PMCW) waveform. In contrast to dual-function radar-communication designs that require an additional modulation layer~\cite{VaeziCOMST2026}, this system uses a maximum-length sequence (MLS) to provide the sharp periodic autocorrelation needed for range processing, while binary data are conveyed by phase-shift keying entire MLS pulses. At the radar, the same surface that forms the coherent communication cascade also inevitably produces singly RIS-assisted, RIS self-, and doubly RIS-assisted echoes, each with its own delay, Doppler shift, and residual phase. Previous research has predominantly treated dual-function waveform design and RIS phase optimization in isolation. Consequently, the effect of a purely communication-optimized surface on the integrated sensing function remains poorly understood, especially in a cluttered environment where static scatterers near the radar compete with the target echo. Building on the communication-optimized (CO) RIS-ISAC system of~\cite{Triwidyastuti2026arxiv}, we analyze its ergodic capacity and, for sensing, the average signal-to-clutter-plus-noise ratio (SCNR) at the direct range-Doppler cell under both RIS-induced and environmental clutter. Main contributions of the present study can be summarized as follows:

\begin{itemize}
    \item We derive the UE ergodic capacity over Nakagami-$m$ fading: an exact closed form for the no-RIS baseline (Corollary~\ref{cor_capacity_direct}), a moment-matched Gamma approximation for the RIS-assisted link that tightens as $L$ increases (Corollary~\ref{cor_capacity_ris_approx}), and a closed-form Jensen upper bound (Lemma~\ref{lem_jensen_bound}). All are Gaussian-input capacities, so $C\leq\min(\Rsym,\,C_{\mathrm{erg}})$ for the underlying BPSK waveform.

    \item For sensing, we derive a closed-form average SCNR at the direct range-Doppler cell. To this end, we obtain the average powers of the four RIS-related returns and the $N_{\mathrm{scat}}$ scatterer returns $\mathsf{S}_i$, which scale as $\Theta(1)$, $\Theta(L)$, $\Theta(L)$, $\Theta(L^2)$, and $\Theta(1)$, respectively. We then weigh them by the MLS range and slow-time window responses to characterize the interference entering the SCNR, thereby separating RIS-induced clutter, which grows with $L$, from geometry-governed environmental clutter.
    This is followed by running Monte Carlo simulations to validate the analysis.

    \item We evaluate range and velocity using the resolution-normalized RMSE and the probability of a correct estimate within a prescribed fraction of a resolution cell---metrics that distinguish correct sub-cell estimates from errors that are significant relative to the radar resolution.
\end{itemize}

\section{CO-RIS-ISAC System and Signal Model}
\label{sec_system}

Our setting is the communication-optimized RIS-assisted ISAC architecture previously proposed in~\cite{Triwidyastuti2026arxiv}: one monostatic PMCW radar simultaneously senses a moving point target and serves a UE, while the RIS phases are chosen solely for the communication link. The new element is the environment: the radar operates amid $N_{\mathrm{scat}}$ discrete static scatterers near the radar that reflect the emission directly back to the receiver.

\subsection{Network Topology, RIS Model, and Waveform Descriptions}

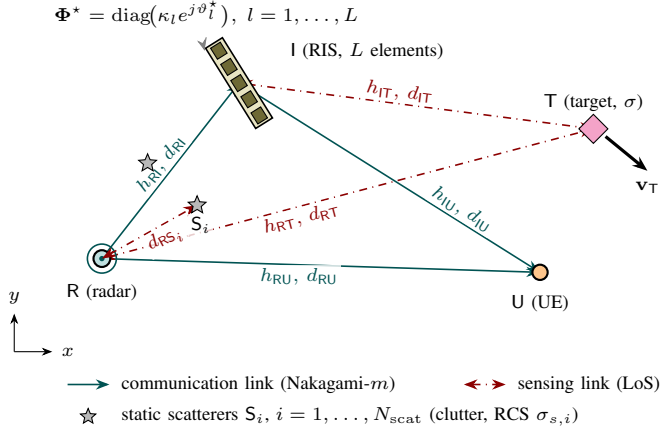
\begin{figure}[t]
    \centering
    \resizebox{\columnwidth}{!}{%
    \begin{tikzpicture}[
        font=\footnotesize,
        comm/.style={-{Stealth[length=1.8mm]},solid,semithick,teal!65!black},
        sens/.style={{Stealth[length=1.8mm]}-{Stealth[length=1.8mm]},dash dot,semithick,red!55!black},
        lbl/.style={font=\scriptsize,inner sep=1pt,fill=white,fill opacity=.8,text opacity=1}
    ]
        \coordinate (R) at (0.55,0.60);
        \coordinate (I) at (2.35,2.90);
        \coordinate (U) at (6.30,0.42);
        \coordinate (T) at (7.00,2.30);

        \draw[comm] (R) -- node[lbl,below,pos=.45]        {$h_{\mathsf{RU}},\,d_{\mathsf{RU}}$} (U);
        \draw[comm] (R) -- node[lbl,sloped,above,pos=.50] {$h_{\mathsf{RI}},\,d_{\mathsf{RI}}$} (I);
        \draw[comm] (I) -- node[lbl,sloped,above,pos=.72] {$h_{\mathsf{IU}},\,d_{\mathsf{IU}}$} (U);
        \draw[sens] (R) -- node[lbl,sloped,below,pos=.40] {$h_{\mathsf{RT}},\,d_{\mathsf{RT}}$} (T);
        \draw[sens] (I) -- node[lbl,sloped,above,pos=.45] {$h_{\mathsf{IT}},\,d_{\mathsf{IT}}$} (T);

        \draw[thick,fill=teal!25] (R) circle (0.11);
        \draw[semithick,teal!65!black] (R) circle (0.19);
        \fill[teal!65!black] (R) circle (0.035);
        \node[below=6pt,font=\scriptsize] at (R) {$\mathsf{R}$ (radar)};

        \begin{scope}[shift={(I)},rotate=-58]
            \draw[fill=olive!20,draw=black,semithick] (-0.62,-0.13) rectangle (0.62,0.13);
            \foreach \x in {-0.50,-0.25,0,0.25,0.50}
                \draw[fill=olive!60!black] (\x-0.085,-0.075) rectangle (\x+0.085,0.075);
        \end{scope}
        \node[font=\scriptsize,anchor=west] at ($(I)+(0.55,0.42)$) {$\mathsf{I}$ (RIS, $L$ elements)};

        \draw[thick,fill=orange!45] (U) circle (0.10);
        \node[below=5pt,font=\scriptsize] at (U) {$\mathsf{U}$ (UE)};

        \node[diamond,draw=black,fill=magenta!45,minimum size=11pt,inner sep=0pt] at (T) {};
        \draw[-{Stealth[length=2mm]},very thick] ($(T)+(0.16,-0.10)$) -- ++(0.55,-0.45)
              node[below,font=\scriptsize]{$\mathbf v_{\mathsf T}$};
        \node[above=3pt,font=\scriptsize] at (T) {$\mathsf{T}$ (target, $\sigma$)};

        \node[anchor=west,font=\scriptsize] (cfg) at (-0.20,3.80)
              {$\boldsymbol{\Phi}^{\star}=\mathrm{diag}\!\big(\kappa_l e^{j\vartheta_l^{\star}}\big),\ l=1,\dots,L$};
        \draw[-{Stealth[length=1.6mm]},gray,shorten >=3pt] (cfg.south) to[bend right=14] ($(I)+(-0.42,0.32)$);

        \coordinate (S1) at (1.80,1.30);
        \coordinate (S2) at (1.15,1.85);
        \draw[sens] (R) -- node[lbl,sloped,below,pos=.60] {$d_{\mathsf{RS}_i}$} (S1);
        \node[star,star points=5,star point ratio=2.2,draw=black,fill=gray!55,minimum size=7pt,inner sep=0pt] at (S1) {};
        \node[star,star points=5,star point ratio=2.2,draw=black,fill=gray!55,minimum size=7pt,inner sep=0pt] at (S2) {};
        \node[lbl] at ($(S1)+(0.05,-0.26)$) {$\mathsf{S}_i$};

        \begin{scope}[shift={(-0.60,-0.62)}]
            \draw[-{Stealth[length=1.4mm]}] (0,0) -- (0.50,0) node[right,font=\scriptsize]{$x$};
            \draw[-{Stealth[length=1.4mm]}] (0,0) -- (0,0.50) node[above,font=\scriptsize]{$y$};
        \end{scope}

        \begin{scope}[shift={(0.10,-1.05)}]
            \draw[comm] (0,0) -- (0.55,0);
            \node[right=1pt,font=\scriptsize] at (0.55,0){communication link (Nakagami-$m$)};
            \draw[sens] (5.20,0) -- (5.75,0);
            \node[right=1pt,font=\scriptsize] at (5.75,0){sensing link (LoS)};
            \node[star,star points=5,star point ratio=2.2,draw=black,fill=gray!55,minimum size=6pt,inner sep=0pt] at (0.27,-0.42) {};
            \node[right=1pt,font=\scriptsize] at (0.55,-0.42){static scatterers $\mathsf{S}_i$, $i=1,\ldots,N_{\mathrm{scat}}$ (clutter, RCS $\sigma_{s,i}$)};
        \end{scope}
    \end{tikzpicture}%
    }
    \caption{Network topology (top view): static scatterers $\mathsf{S}_i$ near the radar reflect the emission back to $\mathsf R$ only.}
    \label{fig_topology}
\end{figure}

Fig.~\ref{fig_topology} shows the planar deployment: a monostatic radar $\mathsf{R}$, a UE $\mathsf{U}$, a point target $\mathsf{T}$ moving with velocity $\mathbf v_{\mathsf T}$, a RIS $\mathsf{I}$, and $N_{\mathrm{scat}}$ stationary point reflectors $\mathsf{S}_i$, $i=1,\ldots,N_{\mathrm{scat}}$, near the radar that constitute environmental clutter. Furthermore, $d_{\mathsf{AB}}$ denotes the $\mathsf A$-$\mathsf B$ distance. The RIS comprises $L$ passive elements that apply reflection coefficients $\theta_l = \kappa_l e^{j\vartheta_l}$, with amplitude $\kappa_l\in(0,1]$ and controllable phase $\vartheta_l\in[0,2\pi)$, $l=1,\ldots,L$, collected in $\boldsymbol{\Phi}=\operatorname{diag}([\theta_1,\ldots,\theta_L])$.

The emission is the MLS-based PMCW waveform previously reported in~\cite{Triwidyastuti2026arxiv}. A bipolar MLS $\mathbf s\in\{\pm1\}^N$ of length $N=2^{k_{\mathrm{MLS}}}-1$ ($k_{\mathrm{MLS}}$ is the MLS order, distinct from the range/delay bin index $k$ in Section~\ref{sec_receiver}) repeats periodically. Each of the $M$ information bits $b_m\in\{0,1\}$ in a coherent processing interval (CPI) is carried by one complete MLS pulse through the BPSK symbol $d_m=1-2b_m$; therefore, sensing and communication share a single emission. The chip duration is $T_c=1/B$ (for unit-energy rectangular NRZ chips with chip-matched filtering, where $B$ is the nominal chip-rate bandwidth), the chip energy is $E_c$, the PRI is $T_{\mathrm{PRI}}=NT_c$, and, because one bit occupies all $N$ chips, the bit energy is $E_b=NE_c$. The frame structure and range/Doppler resolutions follow~\cite{Triwidyastuti2026arxiv}.

\subsection{Channel Models}

\emph{Communication links:} Every communication link between nodes $\mathsf A$ and $\mathsf B$ follows the composite small-/large-scale model~\cite[Eq.~(3.13)]{Goldsmith2005}:
\begin{align}
    h_{\mathsf{AB}}
    =
    \sqrt{\beta_{\mathsf{AB}}}\,a_{\mathsf{AB}}e^{j\phi_{\mathsf{AB}}},
    \quad
    \beta_{\mathsf{AB}}=\Psi_{\mathsf{AB}}(d_{\mathsf{AB}}/d_0)^{-\alpha_{\mathsf{AB}}},
    \label{eq_comm_channel}
\end{align}
where $\beta_{\mathsf{AB}}$ is the large-scale gain, $a_{\mathsf{AB}}\sim\operatorname{Nakagami}(m_{\mathsf{AB}},1)$ is the Nakagami fading amplitude, and $\phi_{\mathsf{AB}}\sim\mathcal U[0,2\pi)$ is the uniform fading phase. This model applies to $h_{\mathsf{RU}}$ and, for each element $l$, to $h_{\mathsf{RI},l}$ and $h_{\mathsf{IU},l}$; in the far field, $\beta_{\mathsf{RI},l}\!\simeq\!\beta_{\mathsf{RI}}$ and $\beta_{\mathsf{IU},l}\!\simeq\!\beta_{\mathsf{IU}}$.

\emph{Sensing links:} The radar-target and RIS-target links are deterministic line-of-sight channels obeying the radar range equation~\cite{Triwidyastuti2026arxiv,Richards2014}; the target's fluctuating scattering is carried by the separate factor $\rho_{\mathsf T}=\sqrt{\sigma}\,e^{j\phi_\sigma}$, with $\sigma$ being the radar cross section (RCS)~\cite[cf.~Eq.~(2.25)]{Goldsmith2005}. With the one-way coefficients $h_{\mathsf{RT}}=\sqrt{\beta_{\mathsf{RT}}}\,e^{-j\frac{2\pi}{\lambda}d_{\mathsf{RT}}}$ and, per RIS element, $h_{\mathsf{IT},l}=\sqrt{\beta_{\mathsf{IT}}}\,e^{-j\frac{2\pi}{\lambda}d_{\mathsf{IT},l}}$ ($\beta_{\mathsf{IT},l}\simeq\beta_{\mathsf{IT}}$ in the far field), reciprocity yields the monostatic direct-target coefficient as follows:
\begin{align}
    \xi_0 = \rho_{\mathsf T}h_{\mathsf{RT}}^2 = \rho_{\mathsf T}\beta_{\mathsf{RT}}e^{-j\frac{4\pi}{\lambda}d_{\mathsf{RT}}},
    \qquad
    |\xi_0|^2=\sigma\beta_{\mathsf{RT}}^2,
    \label{eq_direct_target_coefficient}
\end{align}
with $\beta_{\mathsf{RT}}$ following the same log-distance model as $\beta_{\mathsf{AB}}$ above~\cite{Triwidyastuti2026arxiv}.

\emph{Scatterer links:} Analogously, each static scatterer $\mathsf S_i$ is observed through the deterministic LoS coefficient $h_{\mathsf{RS}_i}=\sqrt{\beta_{\mathsf{RS}_i}}\,e^{-j\frac{2\pi}{\lambda}d_{\mathsf{RS}_i}}$, where $\beta_{\mathsf{RS}_i}$ follows the same log-distance model and $e^{-j\frac{2\pi}{\lambda}d_{\mathsf{RS}_i}}$ accounts for propagation over the distance $d_{\mathsf{RS}_i}$~\cite[cf.~Eq.~(2.6)]{Goldsmith2005}. Its scattering is represented by $\rho_{s,i}=\sqrt{\sigma_{s,i}}\,e^{j\phi_{s,i}}$, where $\sigma_{s,i}$ is the radar cross section of $\mathsf S_i$ and $\phi_{s,i}\sim\mathcal U[0,2\pi)$ is its scattering phase.

\subsection{Communication-Optimized RIS Configuration}

Following~\cite{Triwidyastuti2026arxiv}, the surface is configured for the UE alone, each phase being set as follows:
\begin{align}
    \vartheta_l^\star
    =
    \phi_{\mathsf{RU}}
    -\phi_{\mathsf{RI},l}
    -\phi_{\mathsf{IU},l}
    \pmod{2\pi},
    \quad
    \theta_l^\star=\kappa_l e^{j\vartheta_l^\star}.
    \label{eq_comm_ris_phase}
\end{align}
According to~\cite[Lemma~1]{Triwidyastuti2026arxiv}, the choice in Eq.~\eqref{eq_comm_ris_phase} rotates all cascaded paths onto the phase of the direct path. Thus, the effective channel $h_{\mathrm{eff}}=h_{\mathsf{RU}}+\sum_{l=1}^{L}h_{\mathsf{IU},l}\theta_l^\star h_{\mathsf{RI},l}$ factors as $h_{\mathrm{eff}}=e^{j\phi_{\mathsf{RU}}}A_c$, where the amplitude is real and nonnegative:
\begin{align}
    A_c
    \triangleq
    |h_{\mathrm{eff}}|
    =
    \sqrt{\beta_{\mathsf{RU}}}\,a_{\mathsf{RU}}
    +
    \sqrt{\beta_{\mathsf{RI}}\beta_{\mathsf{IU}}}
    \textstyle\sum_{l=1}^{L}
    \kappa_l\,a_{\mathsf{RI},l}a_{\mathsf{IU},l}.
    \label{eq_effective_amplitude}
\end{align}

After unit-energy chip-matched filtering, the noise at the UE is i.i.d.\ $\mathcal{CN}(0,N_0)$, while the channels are block fading (constant over one CPI). The UE chip-rate observation is unchanged from~\cite{Triwidyastuti2026arxiv}: the direct and RIS-assisted components arrive lumped at the common direct-path delay bin $\ell_u$, scaled by $\sqrt{E_c}\,h_{\mathrm{eff}}$ and spread by the data-modulated MLS.

\subsection{Composite Radar Return}
\label{subsec_radar_return}

According to~\cite[Lemma~2]{Triwidyastuti2026arxiv}, under channel reciprocity and the stop-and-hop approximation, the monostatic radar collects four RIS-related returns indexed by $i=0,\ldots,3$ (the direct, singly RIS-assisted, static RIS self-, and doubly RIS-assisted returns), to which the cluttered environment adds $N_{\mathrm{scat}}$ scatterer returns as follows:
\begin{align}
    &r_{\mathrm{rad}}[n,\!m]
    \!\!=\!\!
    \sqrt{E_c} \!
    \sum_{i=0}^{3} \!
    \xi_i\,
    s[(n \!-\! \ell_i \!+\! 1)\!\bmod\! N]\,d_m
    e^{j2\pi f_{D,i}mT_{\mathrm{PRI}}}
    \nonumber\\
    & +
    \sqrt{E_c}\!
    \sum_{i=1}^{N_{\mathrm{scat}}}\!
    \xi_{s,i}\,
    s[(n-\ell_{s,i}+1)\!\bmod\! N]\,d_m
    +
    z_{\mathrm{rad}}[n,m],
    \label{eq_radar_echo}
\end{align}
where $z_{\mathrm{rad}}[n,m]$ is i.i.d.\ $\mathcal{CN}(0,N_0)$ noise and, with the one-way RIS-assisted cascade $q_{\mathsf{RIT}}=\sum_l\theta_l^\star h_{\mathsf{RI},l}h_{\mathsf{IT},l}$, the effective coefficients are then:
\begin{align}
    \xi_0 &= \rho_{\mathsf T}\beta_{\mathsf{RT}}e^{-j\frac{4\pi}{\lambda}d_{\mathsf{RT}}},
    & \xi_1 &= 2\rho_{\mathsf T}h_{\mathsf{RT}}q_{\mathsf{RIT}},\nonumber\\
    \xi_2 &= \textstyle\sum_l\theta_l^\star h_{\mathsf{RI},l}^2,
    & \xi_3 &= \rho_{\mathsf T}q_{\mathsf{RIT}}^2.
    \label{eq_four_returns}
\end{align}
Each return $i$ arrives at its geometric round-trip delay $\tau_i$ (e.g., $\tau_0=2d_{\mathsf{RT}}/c_0$), which is mapped to chip-lag bin $\ell_i$, and has a Doppler shift $f_{D,i}$ determined by the target's radial velocities $v_r^X=\mathbf v_{\mathsf T}^{\mathsf T}\mathbf u_{X\rightarrow\mathsf T}$ along the $X\!\to\!\mathsf T$ directions, $X\in\{\mathsf R,\mathsf I\}$. In particular, $f_{D,0}=-2v_r^{\mathsf R}/\lambda$, and the RIS self-return is static, with $f_{D,2}=0$. The remaining delays and Doppler shifts, which are not explicitly needed below, are listed in~\cite[Lemma~2]{Triwidyastuti2026arxiv}.

\emph{Environmental clutter:} By reciprocity of the scatterer link, each static scatterer $\mathsf S_i$ contributes the monostatic coefficient
\begin{align}
    \xi_{s,i}
    \!=\!
    \rho_{s,i}h_{\mathsf{RS}_i}^2
    \!\!=\!
    \rho_{s,i}\,\beta_{\mathsf{RS}_i}\,e^{-j\frac{4\pi}{\lambda}d_{\mathsf{RS}_i}},\quad
    |\xi_{s,i}|^2 \!=\! \sigma_{s,i}\beta_{\mathsf{RS}_i}^2,
    \label{eq_scatterer_coefficient}
\end{align}
with delay $\tau_{s,i}=\tfrac{2d_{\mathsf{RS}_i}}{c_0}$ (mapped to the chip-lag bin $\ell_{s,i}$) and zero Doppler shift, since the scatterers are stationary. The scatterers are assumed to interact only with the radar emission: the second-order RIS-scatterer and scatterer-target paths are neglected, and the scatterers do not affect the UE observation, so the communication model of~\cite{Triwidyastuti2026arxiv} and the RIS configuration in Eq.~\eqref{eq_comm_ris_phase} remain unchanged.

\section{Receiver Processing}
\label{sec_receiver}

\subsection{Sensing Function: Range-Doppler Processing}
\label{subsec_sensing_processing}

As in~\cite{Triwidyastuti2026arxiv}, the radar first removes its own BPSK symbols, $y[n,m]=d_m^{*}r_{\mathrm{rad}}[n,m]$ (the noise statistics are unchanged because $|d_m|=1$); then circularly correlates each pulse with the MLS as follows:
\begin{align}
    g[k,m] = \textstyle \frac{1}{N}  \sum_{n=1}^{N} y[n,m]\,s[(n-k+1)\bmod N].
    \label{eq_range_corr}
\end{align}
As the normalized MLS periodic autocorrelation equals $R_s[\delta]=1$ at $\delta=0$ and $-1/N$ elsewhere, a return at delay $\ell_i$ peaks at bin $k=\ell_i$, off-peak lags are suppressed, and the post-correlation noise variance is $N_0/N$. Bin $k$ maps to range $R_k=c_0(k-1)T_c/2$ ($k=1$ at zero range), with spacing $\Delta R=c_0T_c/2$. An optional slow-time window $w[m]$ and a coherent-gain-normalized $M$-point DFT for each range bin then yield the following range-Doppler map:
\begin{align}
    G[k,p]
    = \textstyle
    \frac{1}{\sum_{m=1}^M w[m]}\sum_{m=1}^M w[m]\,g[k,m]\,e^{-j2\pi pm/M},
    \label{eq_rd_map}
\end{align}
with Doppler resolution $\Delta f_D=1/(MT_{\mathrm{PRI}})$ and power $P_{\mathrm{RD}}[k,p]=|G[k,p]|^2$. Of the returns in Eq.~\eqref{eq_radar_echo}, only the direct return $\xi_0$ is unaffected by RIS phase misalignment; under the phases in Eq.~\eqref{eq_comm_ris_phase}, the RIS-assisted returns $\xi_1,\xi_2,\xi_3$ are mis-phased across elements and appear as clutter in $P_{\mathrm{RD}}$~\cite{Triwidyastuti2026arxiv}. The scatterer returns $\xi_{s,i}$ peak at their range bins $\ell_{s,i}$ on the zero-Doppler line, merging with the static RIS self-return $\xi_2$ into a common clutter ridge (see Section~\ref{sec_results} and Fig.~\ref{fig_rdmap}).

Targets are declared by a two-dimensional cell-averaging CFAR (CA-CFAR) detector~\cite{Barkat2005}. A local background level $\widehat Z[k,p]$ is formed by averaging $N_{\mathrm{tr}}$ training cells in a ring around each cell under test, excluding a guard region with range and Doppler half-widths $G_r,G_d$ and using training half-widths $T_r,T_d$. A detection is declared if $P_{\mathrm{RD}}[k,p] > \alpha_{\mathrm{CFAR}}\widehat Z[k,p]$ and $[k,p]$ is a local maximum, where $\alpha_{\mathrm{CFAR}} = N_{\mathrm{tr}}(P_{\mathrm{FA}}^{-1/N_{\mathrm{tr}}}-1)$ sets the target false-alarm probability $P_{\mathrm{FA}}$~\cite{Barkat2005} (unrelated to the path-loss exponents $\alpha_{\mathsf{AB}}$). CA-CFAR assumes a locally homogeneous background, so a scatterer peak inside the training ring raises $\widehat Z[k,p]$ and can mask nearby detections. In the considered geometry (i.e., near-radar scatterers on the zero-Doppler line and a distant moving target), the training ring around the direct-return cell remains clutter-free. The detected direct-return bin $(\widehat k,\widehat p)$, where $\widehat p$ is the signed, zero-centered Doppler bin ($p\in\{-\lfloor M/2\rfloor,\ldots,\lceil M/2\rceil-1\}$), is then refined to sub-cell accuracy by separately applying parabolic interpolation to $10\log_{10}P_{\mathrm{RD}}$ around the peak in the range and Doppler dimensions. This yields the range and radial-velocity estimates as follows:
\begin{align}
    \widehat R = \frac{c_0(\widehat k-1)T_c}{2},
    \qquad
    \widehat v_r^{\mathsf R} = -\frac{\lambda\widehat p}{2MT_{\mathrm{PRI}}},
    \label{eq_rd_estimates}
\end{align}
which are then used for the estimation-accuracy metrics of Section~\ref{subsec_sensing_metrics}.

\subsection{Communication Function: Coherent BPSK Detection}
\label{subsec_comm_detection}

The UE acquires the direct-path code delay noncoherently and despreads each pulse with the delay-aligned MLS, exactly as previously described in~\cite{Triwidyastuti2026arxiv}; under correct acquisition ($\widehat\ell_u=\ell_u$), the per-bit observation reduces to
\begin{align}
    u_m
    =
    \sqrt{E_c}\,h_{\mathrm{eff}}\,d_m+\nu_m,
    \qquad
    \nu_m\sim\mathcal{CN}(0,N_0/N).
    \label{eq_um}
\end{align}
The effective channel $h_{\mathrm{eff}}$ is estimated by least squares from an $M_{\mathrm{pl}}$-symbol pilot prefix, and coherent BPSK detection forms $\Lambda_m=\operatorname{Re}\{\widehat h_{\mathrm{eff}}^{*}u_m\}$~\cite{Triwidyastuti2026arxiv}. Since one bit spans all $N$ chips ($E_b=NE_c$), the post-despreading bit SNR at the UE is the quantity governing the ergodic capacity analyzed in Section~\ref{sec_performance_analysis}.

If pilot overhead is included in net throughput, we multiply the data rate by $(M-M_{\mathrm{pl}})/M$.

\section{Performance Analysis}
\label{sec_performance_analysis}

\subsection{Sensing SNR and Estimation Metrics}
\label{subsec_sensing_metrics}

Let $w[m]$ be the slow-time window in Eq.~\eqref{eq_rd_map}, with coherent processing efficiency
$
    L_w
    \triangleq
    \frac{\left|\sum_{m=1}^{M}w[m]\right|^2}
         {M\sum_{m=1}^{M}|w[m]|^2}\in(0,1],
$
where $L_w=1$ for a rectangular window.
For the direct target return, range compression in Eq.~\eqref{eq_range_corr} contributes a coherent gain of $N$, and windowed slow-time integration in Eq.~\eqref{eq_rd_map} contributes a further factor of $ML_w$. Thus, the integrated sensing SNR at an on-grid range-Doppler cell is $\chi_0 = NM L_w\frac{E_c}{N_0}|\xi_0|^2 = NM L_w\frac{E_c}{N_0}\sigma\beta_{\mathsf{RT}}^2$~\cite[Lemma~3]{Triwidyastuti2026arxiv}, with $|\xi_0|^2$ from Eq.~\eqref{eq_direct_target_coefficient}.

The phases in Eq.~\eqref{eq_comm_ris_phase} that combine coherently at the UE via Eq.~\eqref{eq_effective_amplitude} leave the sensing phases of $\xi_1,\xi_2,\xi_3$ in Eq.~\eqref{eq_four_returns} misaligned, so these returns add non-coherently and never enhance $\chi_0$~\cite{Triwidyastuti2026arxiv}: the static self-return $\xi_2$ contributes zero-Doppler clutter, while $\xi_0$ is independent of the RIS. Scatterer returns $\xi_{s,i}$ of Eq.~\eqref{eq_scatterer_coefficient} add further zero-Doppler clutter, independent of the RIS and of $L$. Their impact on direct-target processing is set by their average powers and range-Doppler leakage into the direct cell.

\subsubsection{Average Return Powers}
Let $P_i\triangleq\mathbb E[|\xi_i|^2]$. Under the communication-optimized phases, let us define $v_l \triangleq \kappa_l^2 \beta_{\mathsf{RI},l} \beta_{\mathsf{IT},l}$ and $u_l \triangleq \kappa_l^4 (\beta_{\mathsf{RI},l} \beta_{\mathsf{IT},l})^2 \mathcal M_4(m_{\mathsf{RI}})$, where $\mathcal M_4(m_\mathsf{RI})=\mathbb{E}[a_\mathsf{RI}^4]=(m_\mathsf{RI}+1)/m_\mathsf{RI}$~\cite[cf.~Eq.~(2.169)]{Barkat2005}. The independence of the residual element phases gives the following:
\begin{align*}
    P_0 &= \sigma\beta_{\mathsf{RT}}^2,
    \hspace{2pt} P_1 = 4\sigma\beta_{\mathsf{RT}}\sum_{l=1}^{L}v_l,
    \hspace{2pt} P_2 = \sum_{l=1}^{L}\kappa_l^2\beta_{\mathsf{RI},l}^2
           \mathcal M_4(m_{\mathsf{RI}}), \\
    P_3 & \textstyle = \sigma \left( \sum_{l=1}^{L}u_l
			+2 \left( \left( \sum_{l=1}^L v_l \right)^2 - \sum_{l=1}^L v_l^2 \right) \right).
\end{align*}
In addition, each scatterer return has the following deterministic average power:
\begin{align}
    P_{s,i} \triangleq \mathbb E[|\xi_{s,i}|^2] = \sigma_{s,i}\,\beta_{\mathsf{RS}_i}^2,
    \qquad i=1,\ldots,N_{\mathrm{scat}}.
    \label{eq_scatterer_power}
\end{align}
For identical far-field gains, $P_0=\Theta(1)$, $P_1=\Theta(L)$, $P_2=\Theta(L)$, and $P_3=\Theta(L^2)$, whereas $P_{s,i}=\Theta(1)$ is independent of $L$: environmental clutter is governed purely by geometry. Since the two-way gain $\beta_{\mathsf{RS}_i}^2$ scales as $d_{\mathsf{RS}_i}^{-2\alpha}$, near-radar scatterers of modest RCS can nevertheless dominate the distant target return $P_0$; the expressions are validated against Monte Carlo averages in Section~\ref{sec_results}.

\subsubsection{Direct-Cell Interference}
Let $A_{R,i}$ and $A_{D,i}$ denote the normalized MLS range response and slow-time window response of return $i$, respectively, evaluated at the direct-return cell $(k_0,p_0)$. Furthermore, let $A_{R,s,i}$ and $A_{D,s,i}$ be the corresponding responses of scatterer return $i$. The average RIS- and clutter-induced interference at that cell is then computed as follows:
\begin{align}
    P_{\mathrm{int},0}
    =& \textstyle E_c\sum_{i=1}^{3}P_i|A_{R,i}|^2|A_{D,i}|^2
    \nonumber\\
    & \textstyle +E_c\sum_{i=1}^{N_{\mathrm{scat}}}P_{s,i}|A_{R,s,i}|^2|A_{D,s,i}|^2.
    \label{eq_direct_cell_interference}
\end{align}
The scatterer terms take a particularly simple form. Because the scatterers lie near the radar while the target is distant, $\ell_{s,i}\neq\ell_0$, and the MLS periodic autocorrelation gives exactly $|A_{R,s,i}|^2=1/N^2$. Furthermore, because the scatterers are stationary, $A_{D,s,i}$ is the window response at the target Doppler bin, which is the same zero-Doppler leakage factor as that of the RIS self-return $\xi_2$. Therefore, all static clutter (the RIS self-return and scatterers) enters through a common zero-Doppler leakage mechanism. The average output SCNR is thus
\begin{align}
    \textstyle \overline\chi_{\mathrm{SCNR},0}
    =\frac{E_cP_0}
    {N_0/(NM L_w)+P_{\mathrm{int},0}}.
\end{align}
Accordingly, a strong RIS or scatterer return causes little direct-cell interference when it is sufficiently separated in range or Doppler, whereas nearly coincident returns reduce the SCNR.

\subsubsection{Estimation Accuracy Metrics}
\label{subsubsec_estimation_metrics}

Beyond detection, we evaluate the \emph{correctness} of the range and velocity estimates relative to the resolution cell, rather than the raw error magnitude. Let $\widehat R,\widehat v$ be the sub-cell estimates in Eq.~\eqref{eq_rd_estimates} (with $v\triangleq v_r^{\mathsf R}$ for brevity), and let $\Delta R=c_0T_c/2$ and $\Delta v=\lambda/(2MT_{\mathrm{PRI}})$ be the range and velocity resolution-cell sizes, respectively. The normalized RMSE is
\begin{align}
    \mathrm{NRMSE}_R \!=\! \frac{\sqrt{\mathbb E[(\widehat R \!-\! R)^2]}}{\Delta R},
    \quad
    \mathrm{NRMSE}_v \!=\! \frac{\sqrt{\mathbb E[(\widehat v \!-\! v)^2]}}{\Delta v},
    \label{eq_nrmse}
\end{align}
where $\mathrm{NRMSE}<1$ indicates sub-resolution (i.e., correct) accuracy. The probabilities of correct range and velocity estimates at tolerance $\varepsilon$ (in resolution cells) are
\begin{align}
    \Pr\nolimits_R(\varepsilon) &= \Pr\!\big(|\widehat R-R|\le\varepsilon\Delta R\big),\label{eq_pcorr_R}
    \\
    \Pr\nolimits_v(\varepsilon) &= \Pr\!\big(|\widehat v-v|\le\varepsilon\Delta v\big).
    \label{eq_pcorr_v}
\end{align}
These probabilities are evaluated at a half-cell tolerance $\varepsilon=1/2$ and a tight sub-bin tolerance $\varepsilon=0.1$. Eq.~\eqref{eq_nrmse} and \eqref{eq_pcorr_R}--\eqref{eq_pcorr_v}, evaluated by Monte Carlo simulation using the estimator in Eq.~\eqref{eq_rd_estimates} of Section~\ref{subsec_sensing_processing}, are our main sensing performance metrics (see Section~\ref{sec_results}, Figs.~\ref{fig_nrmse} and~\ref{fig_pcorr}).

\subsection{Communication SNR and Reused Statistical Model}
\label{subsec_snr_amplitude}

The post-despreading bit SNR at the UE is computed as follows:
\begin{align}
    \textstyle \gamma_b
    &=
    \frac{E_b|h_{\mathrm{eff}}|^2}{N_0}
    =
    \overline\gamma A_c^2,
    \qquad
    \overline\gamma \triangleq\frac{E_b}{N_0} =N\frac{E_c}{N_0},
    \label{eq_bit_snr}
\end{align}
where $A_c$ is the effective amplitude in Eq.~\eqref{eq_effective_amplitude}. Write $A_c=X_0+\sum_{l=1}^{L}X_l$, where the direct component $X_0\triangleq\sqrt{\beta_{\mathsf{RU}}}\,a_{\mathsf{RU}}$ and the RIS-assisted components $X_l\triangleq\kappa_l\sqrt{\beta_{\mathsf{RI}}\beta_{\mathsf{IU}}}\,a_{\mathsf{RI},l}a_{\mathsf{IU},l}$ are mutually independent. Their component moments $\mu_{i,r}\triangleq\mathbb E[X_i^r]$, $i=0,\ldots,L$, are available in closed form from the Nakagami raw moments~\cite[Eq.~(2.169)]{Barkat2005}. Then, the raw amplitude moments $\mathcal A_{L,r}\triangleq\mathbb E[A_c^r]$, $r\leq4$, follow exactly from the amplitude-moment recursion in~\cite[Lemma~4]{Triwidyastuti2026arxiv}. Matching the mean and variance of $\gamma_b=\overline\gamma A_c^2$ to a Gamma law gives the moment-matched model of~\cite{Triwidyastuti2026arxiv},
\begin{align}
    \textstyle \gamma_b\mathrel{\dot\sim}\operatorname{Gamma}(\zeta_\gamma,\eta_\gamma),
    \quad
    \zeta_\gamma=\frac{\mu_\gamma^2}{\sigma_\gamma^2},
    \quad
    \eta_\gamma=\frac{\sigma_\gamma^2}{\mu_\gamma},
    \label{eq_gamma_fit}
\end{align}
where $\mathrel{\dot\sim}$ denotes approximation in distribution, $\mu_\gamma=\overline\gamma\,\mathcal A_{L,2}$, and $\sigma_\gamma^2=\overline\gamma^2(\mathcal A_{L,4}-\mathcal A_{L,2}^2)$; the fit is not exact for $L>0$, but tightens as $L$ grows and $A_c$ aggregates more independent terms. For the no-RIS case, from~\cite{Triwidyastuti2026arxiv}, we obtain the following:
\begin{align}
    \textstyle \gamma_{b,0}=\overline\gamma\beta_{\mathsf{RU}}a_{\mathsf{RU}}^2
    \sim\operatorname{Gamma}\!\Big(m_{\mathsf{RU}},\tfrac{\overline\gamma\beta_{\mathsf{RU}}}{m_{\mathsf{RU}}}\Big).
    \label{eq_noris_dist}
\end{align}
The statistical models in Eq.~\eqref{eq_gamma_fit} and \eqref{eq_noris_dist} are reused from~\cite{Triwidyastuti2026arxiv}; the subsequent EC analysis is the contribution of the present paper.

\subsection{Ergodic Capacity Analysis}
\label{subsec_capacity_def}

Each despread observation in Eq.~\eqref{eq_um} is one channel use per PRI, so the symbol rate is $\Rsym=1/T_{\mathrm{PRI}}=B/N$, not the chip rate $B=N\Rsym$, which would overcount by the spreading factor $N$. With fixed power and rate and receiver-only CSI (the pilot-based $\widehat h_{\mathrm{eff}}$ of Section~\ref{sec_receiver}), the per-state rate $\Rsym\log_2(1+\gamma_b)$~\cite[Eq.~(5.12)]{Tse2005} is \emph{not} achievable moment to moment~\cite[Section~4.2.3]{Goldsmith2005}; the relevant measure is ergodic capacity $C_{\mathrm{erg}} = \Rsym\,\mathbb E[\log_2(1 + \gamma_b)]$ [bits/s]~\cite[Eq.~(4.4)]{Goldsmith2005}, with the expectation over the bit-SNR distribution in Eq.~\eqref{eq_bit_snr}. As $C_{\mathrm{erg}}$ is a Gaussian-input quantity, whereas BPSK carries at most one bit per channel use, it is an idealized upper bound: the true capacity satisfies $C\leq\min(\Rsym,\,C_{\mathrm{erg}})$.

For the direct (no-RIS) link, the exact Gamma distribution in Eq.~\eqref{eq_noris_dist} admits an exact closed-form ergodic capacity expression.

\begin{lemma}[Capacity for Gamma-distributed SNR]
\label{lem_capacity_gamma_series}
Let $\gamma_b \sim \operatorname{Gamma}(\zeta, \eta)$ with shape $\zeta$ and scale $\eta$.
For integer $\zeta$, the per-symbol ergodic capacity has the following closed form:
\begin{align}
	&C_{\mathrm{erg,Gamma}}(\zeta,\eta)
	=~\frac{1}{\Gamma(\zeta)\ln 2} \textstyle \sum_{i=0}^{\zeta-1}\frac{(\zeta-1)!}{(\zeta-i-1)!} \nonumber \\
	&\times\biggl[ \frac{(-1)^{\zeta-i-2}}{\eta^{\zeta-i-1}} e^{1/\eta} \operatorname{Ei} \Bigl( -\frac{1}{\eta} \Bigr) +\! \sum_{k=1}^{\zeta-i-1}\frac{(k-1)!}{(-\eta)^{\zeta-i-k-1}}\biggr],
	\label{eq_capacity_ei}
\end{align}
where $\operatorname{Ei}(\cdot)$ is the exponential integral and $\Gamma(\zeta)$ the Gamma function~\cite[Eq.~(2.98)]{Barkat2005}; the summation index $i$ is local to Eq.~\eqref{eq_capacity_ei}, unrelated to the radar-return index of Eq.~\eqref{eq_four_returns}.
\end{lemma}
\begin{proof}
Starting from $C_{\mathrm{erg,Gamma}}=\mathbb E[\log_2(1+\gamma_b)]$ with the Gamma density $f_{\gamma_b}(x) = \frac{x^{\zeta-1} e^{-x/\eta}}{\Gamma(\zeta)\eta^\zeta}$~\cite[Eq.~(2.102)]{Barkat2005}, substitute $u = x/\eta$ to obtain
\begin{equation}
	\textstyle C_{\mathrm{erg,Gamma}}(\zeta,\eta)
	=\frac{1}{\Gamma(\zeta)}\!\int_{0}^{\infty}\!\log_2(1+\eta u)\,
	u^{\zeta-1}e^{-u}\,\mathrm{d}u.
	\label{eq_capacity_integral}
\end{equation}
Applying~\cite[Eq.~(4.337.5)]{Gradshteyn2007} directly to Eq.~\eqref{eq_capacity_integral} for integer $\zeta$ yields the closed form in Eq.~\eqref{eq_capacity_ei}.
\end{proof}

For the direct link, Lemma~\ref{lem_capacity_gamma_series} can be applied to the exact Gamma distribution in Eq.~\eqref{eq_noris_dist}.

\begin{corollary}[Direct-link ergodic capacity (exact)]
\label{cor_capacity_direct}
For the no-RIS case, the ergodic capacity is computed as follows:
\begin{align}
    \textstyle C_{\mathrm{erg},0} = \Rsym\, C_{\mathrm{erg,Gamma}}\!\left(m_{\mathsf{RU}}, \frac{\overline\gamma\beta_{\mathsf{RU}}}{m_{\mathsf{RU}}}\right)\text{ [bits/s]}.
    \label{eq_capacity_direct_exact}
\end{align}
For integer $m_{\mathsf{RU}}$, $C_{\mathrm{erg,Gamma}}$ is evaluated using the closed form in Eq.~\eqref{eq_capacity_ei}; for arbitrary real $m_{\mathsf{RU}}>0$, the exact integral representation in Eq.~\eqref{eq_capacity_integral} applies.
\end{corollary}

For the RIS-assisted case, the SNR $\gamma_b = \overline\gamma A_c^2$ is not exactly Gamma distributed, but is well approximated via the moment-matched Gamma fit in Eq.~\eqref{eq_gamma_fit}:

\begin{corollary}[RIS-assisted ergodic capacity (moment-matched)]
\label{cor_capacity_ris_approx}
Under the communication-optimized RIS configuration, the ergodic capacity is approximated as
\begin{align}
    C_{\mathrm{erg,RIS}} \approx \Rsym\, C_{\mathrm{erg,Gamma}}\!\left(\zeta_\gamma, \eta_\gamma\right)\text{ [bits/s]},
    \label{eq_capacity_ris_approx}
\end{align}
where $\zeta_\gamma, \eta_\gamma$ are from Eq.~\eqref{eq_gamma_fit} and $C_{\mathrm{erg,Gamma}}$ is evaluated using the integral form in Eq.~\eqref{eq_capacity_integral} by adaptive quadrature, since $\zeta_\gamma$ is generally non-integer.
\end{corollary}

While the evaluations in Eq.~\eqref{eq_capacity_ei} and \eqref{eq_capacity_integral} are exact, closed-form bounds provide intuitive insights and require only mean and variance.

\begin{lemma}[Jensen's inequality upper bound]
\label{lem_jensen_bound}
For any SNR distribution with finite mean $\bar\gamma = \mathbb E[\gamma_b]$,
\begin{align}
    C_{\mathrm{erg}} \leq C_{\mathrm{erg,Jensen}} \triangleq \Rsym\log_2(1 + \bar\gamma).
    \label{eq_jensen_bound}
\end{align}
\end{lemma}

\begin{proof}
By Jensen's inequality for the concave $\log_2(1+x)$~\cite[Eq.~(5.14)]{Tse2005}, $\mathbb E[\log_2(1 + \gamma_b)]\leq \log_2(1 + \mathbb E[\gamma_b])$, with equality only for deterministic $\gamma_b$~\cite[Eq.~(5.92)]{Tse2005}.
\end{proof}

For the RIS system with mean $\bar\gamma = \overline\gamma \mathcal A_{L,2}$:
\begin{align}
    C_{\mathrm{erg,Jensen}} = \Rsym\log_2\left(1 + \overline\gamma \mathcal A_{L,2}\right),
    \label{eq_jensen_ris}
\end{align}
with $\mathcal A_{L,2}$ from the amplitude-moment recursion of \cite[Lemma~4]{Triwidyastuti2026arxiv}. Since
$
	\mathcal A_{L,2}=\mathbb E[A_c^2]=\sum_{i=0}^L\mu_{i,2}+2\sum_{0\leq i<j\leq L}\mu_{i,1}\mu_{j,1},
$
the bound in Eq.~\eqref{eq_jensen_ris} depends on the fading only through $\mathbb E[\gamma_b]=\overline\gamma\mathcal A_{L,2}$.

\section{Numerical Results}
\label{sec_results}

Monte Carlo simulations were used to validate the analysis. The key settings are listed in Table~\ref{tab_params}: the communication links underwent Nakagami-$m$ fading, and the RIS applied the communication-optimized phases in Eq.~\eqref{eq_comm_ris_phase}.

\begin{table}[t]
\centering
\footnotesize
\renewcommand{\arraystretch}{1.0}
\caption{Key Simulation Settings}
\label{tab_params}
\begin{tabular}{@{}ll@{}}
\toprule
Parameter & Value \\
\midrule
Carrier frequency $f_c$ (wavelength $\lambda$)         & $24$~GHz ($12.5$~mm)\\
MLS order / length $k_{\mathrm{MLS}}$ / $N$              & $7$ / $127$\\
Pulses per CPI $M$ / pilots $M_{\mathrm{pl}}$           & $64$ / $8$\\
Chip rate / bandwidth $B$                               & $20$~MHz\\
Chip / PRI / CPI duration                              & $50$~ns / $6.35\,\mu$s / $406\,\mu$s\\
Reference distance $d_0$                               & $100$~m\\
Path-loss exp. $\alpha_{\mathsf{RU}}$ / $\alpha_{\mathsf{RI}},\alpha_{\mathsf{IU}}$ & $2.2$ / $2.0$\\
Path-loss exp. $\alpha_{\mathsf{RT}},\alpha_{\mathsf{IT}},\alpha_{\mathsf{RS}}$ & $2.0$\\
Ref. gains $\Psi_{\mathsf{RU}}$ / $\Psi_{\mathsf{RI}},\Psi_{\mathsf{IU}}$ & $0.5$ / $0.02$\\
Ref. gains $\Psi_{\mathsf{RT}}$ / $\Psi_{\mathsf{IT}}$ ($\Psi_{\mathsf{RS}}{=}\Psi_{\mathsf{RT}}$) & $1.58{\times}10^{-5}$ / $4.99{\times}10^{-8}$\\
Nakagami $m_{\mathsf{RU}},m_{\mathsf{RI}},m_{\mathsf{IU}}$ & $2,\,3,\,2$\\
RIS elements $L$ / spacing / $\kappa_l$                & $16,64$ / $\lambda/2$ / $0.85$\\
$\mathsf R,\mathsf I,\mathsf U$ positions               & $(0,0),(200,60),(380,0)$\\
Target $\mathsf T$ pos.\ / vel.\ $\mathbf v_{\mathsf T}$ & $(560,280)$~m / $(28,-18)$~m/s\\
Target RCS $\sigma$ & $1$~m$^2$ \\
Scatterers $N_{\mathrm{scat}}$ / RCS $\sigma_{s,i}$ [m$^2$] & $4$ / $1$\\
$\mathsf S_i$ positions & \makecell[l]{$(28,12),(50,-25)$,\\$(68,30),(82,-8)$}\\
Slow-time window $w[m]$ ($L_w$) & Hann ($\approx\!0.66$)\\
CFAR $T_r,T_d$ / $G_r,G_d$ / $P_{\mathrm{FA}}$ & $4,4$ / $1,1$ / $10^{-4}$\\
\bottomrule
\end{tabular}
\end{table}

\begin{figure}[t]
    \centering
    \includegraphics[width=\columnwidth,height=0.6\columnwidth]{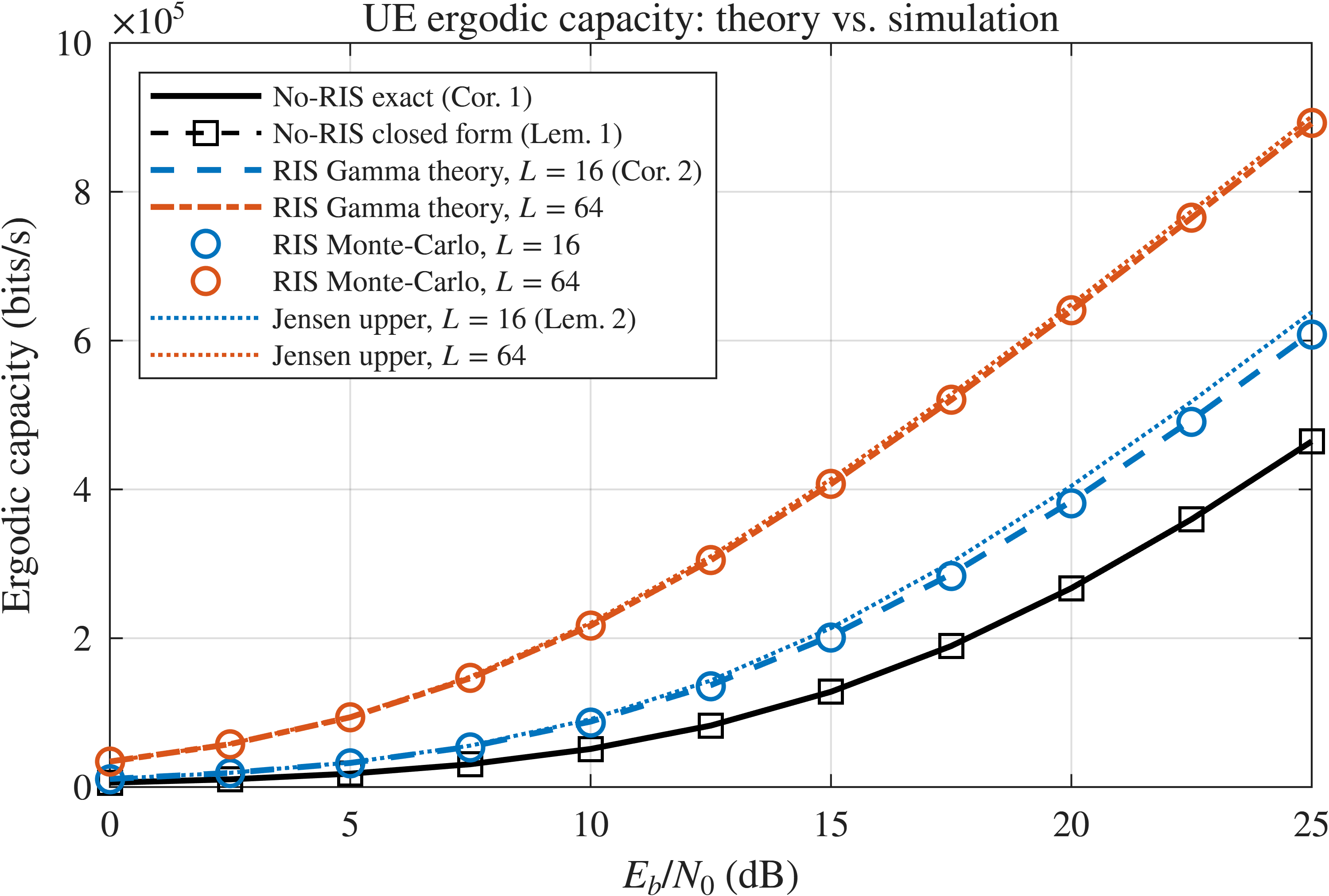}
    \caption{UE ergodic capacity vs. $E_b/N_0$: the no-RIS baseline in Eq.~\eqref{eq_capacity_direct_exact} (line), cross-validated against the closed form in Eq.~\eqref{eq_capacity_ei} (squares); and the RIS-assisted Gamma theory in Eq.~\eqref{eq_capacity_ris_approx} (lines) vs. Monte Carlo simulation (circles) for $L=16,64$, with the Jensen bound in Eq.~\eqref{eq_jensen_bound} (dotted) for both values of $L$.}
    \label{fig_capacity}
\end{figure}

Fig.~\ref{fig_capacity} cross-validates the exact no-RIS baseline in Eq.~\eqref{eq_capacity_direct_exact} against the Ei-based closed form in Eq.~\eqref{eq_capacity_ei} (both are exact evaluations for the same Gamma-distributed no-RIS SNR) and compares the moment-matched Gamma theory in Eq.~\eqref{eq_capacity_ris_approx} with Monte Carlo averages of $\Rsym\log_2(1+\gamma_b)$, where $\Rsym\approx0.16$~Mbit/s. The Gamma theory closely tracked the simulation for both $L=16$ and $L=64$; in both cases, the simulation remained below the Jensen bound in Eq.~\eqref{eq_jensen_bound}.
Fig.~\ref{fig_rdmap} shows the range-Doppler map in Eq.~\eqref{eq_rd_map} for all returns in Eq.~\eqref{eq_radar_echo}. The RIS-assisted echoes did not coherently reinforce the target, and the $L^2$ gain delivered to the UE did not aid sensing. Meanwhile, the near-radar scatterers and the RIS self-return merged into the zero-Doppler clutter ridge anticipated in Section~\ref{sec_receiver}.

\begin{figure}[t]
    \centering
    \includegraphics[width=\columnwidth,height=0.6\columnwidth]{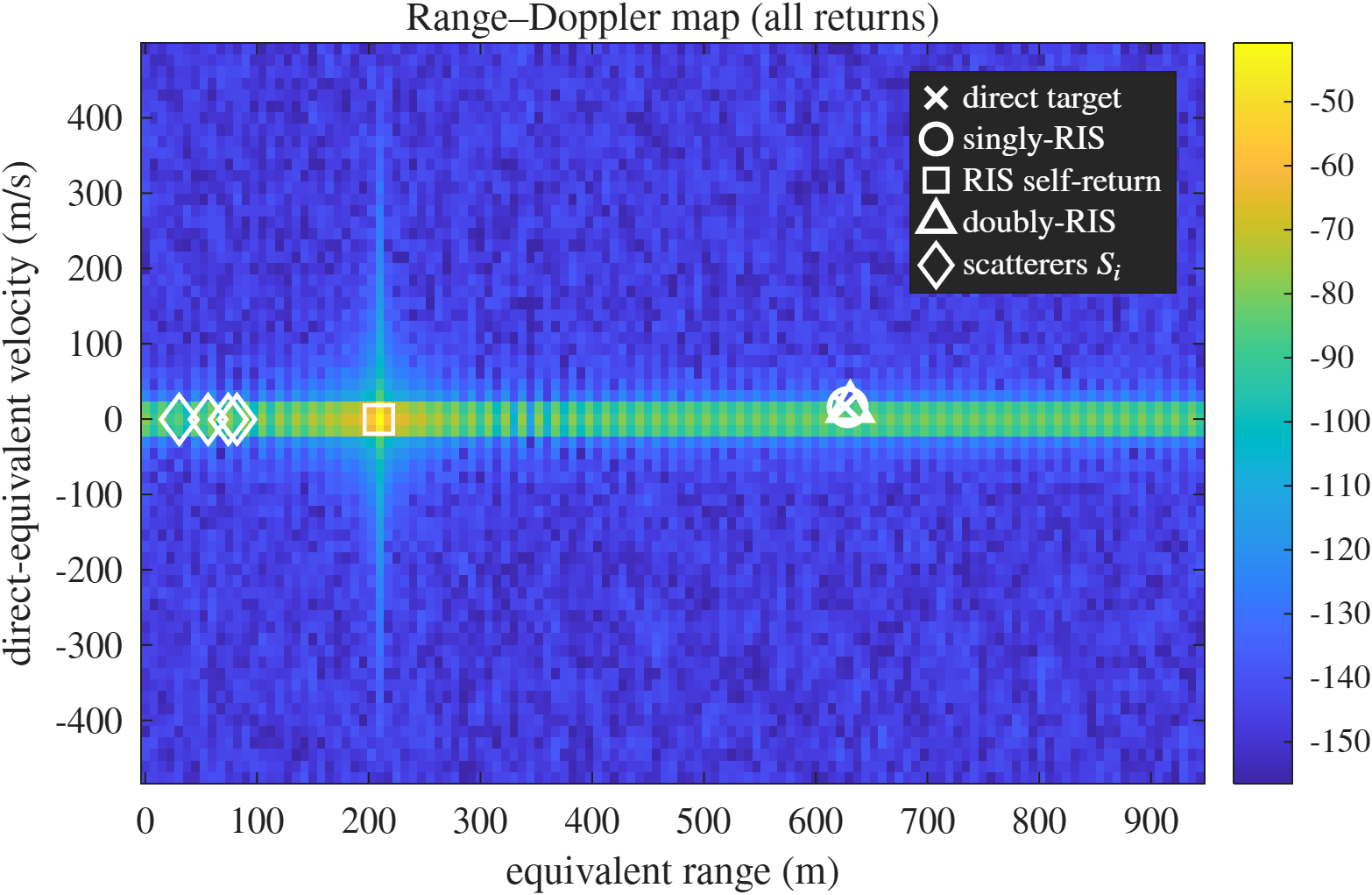}
    \caption{RD map of the composite radar return, with expected return locations marked: the target-related returns cluster together, while the RIS self-return and the near-radar scatterers $\mathsf{S}_i$ form the zero-Doppler clutter ridge.}
    \label{fig_rdmap}
\end{figure}

\begin{figure}[t]
    \centering
    \subfloat[]{\includegraphics[width=0.49\columnwidth]{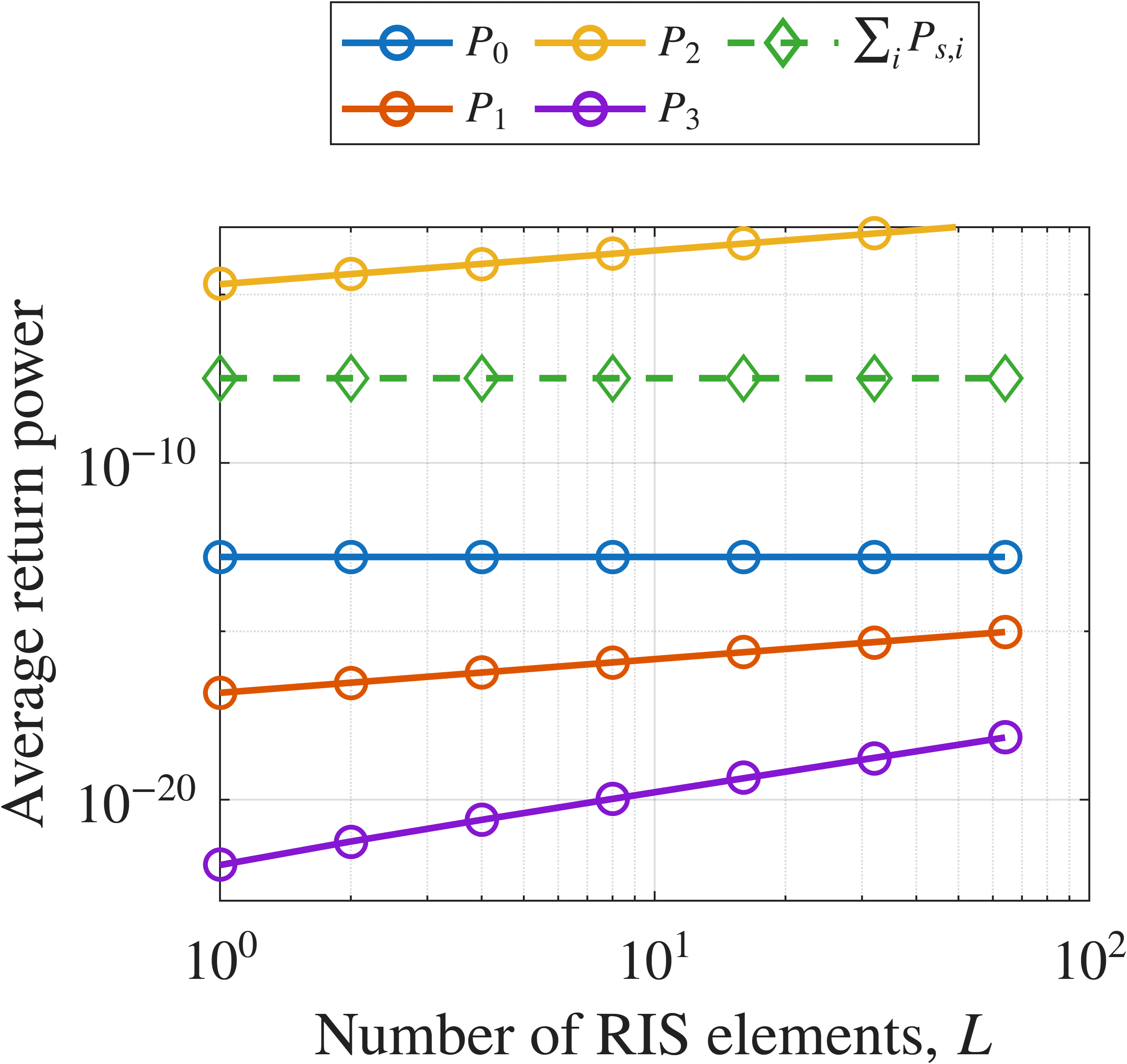}%
        \label{fig_return_power}}
    \hfil
    \subfloat[]{\includegraphics[width=0.49\columnwidth]{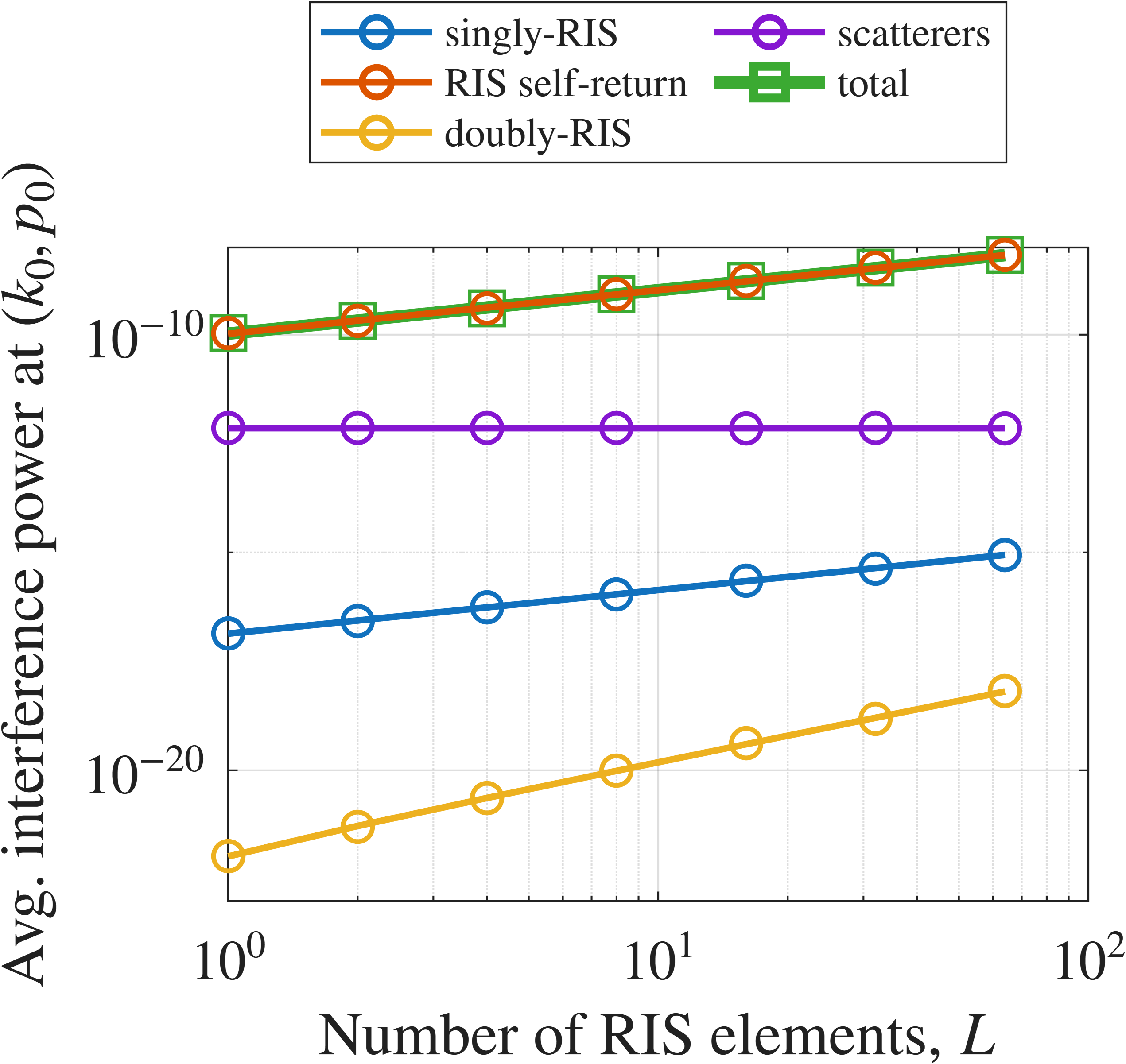}%
        \label{fig_direct_cell_interference}}
    \caption{(a) Average radar-return powers and (b) average interference at the direct range-Doppler cell vs. $L$: analysis (lines) and simulation (markers).}
    \label{fig_power_interference}
\end{figure}

Fig.~\ref{fig_return_power} validates the return-power expressions in Section~\ref{subsec_sensing_metrics}: $P_0$ was independent of $L$; $P_1$ and $P_2$ grew linearly; and $P_3$ grew quadratically. The total scatterer power $\sum_i P_{s,i}$ in Eq.~\eqref{eq_scatterer_power} was constant with respect to $L$ yet exceeded the distant direct-target return by roughly $53$~dB because of the near-radar two-way gain $d_{\mathsf{RS}_i}^{-2\alpha}$.
Fig.~\ref{fig_direct_cell_interference} weights each power by its range and Doppler leakage at the direct cell, thereby separating a return's physical power from the fraction that actually caused interference; the analytical model closely follows the simulation. The zero-Doppler RIS self-return dominated, whereas the much stronger raw scatterer clutter was suppressed by the MLS sidelobe factor $1/N^2$ in Eq.~\eqref{eq_direct_cell_interference} and contributed about $34$~dB less. (Monte Carlo values near $10^{-22}$ are averages of continuous-valued powers and are well within double precision.)

\begin{figure}[t]
    \centering
    \subfloat[]{\includegraphics[width=0.49\columnwidth]{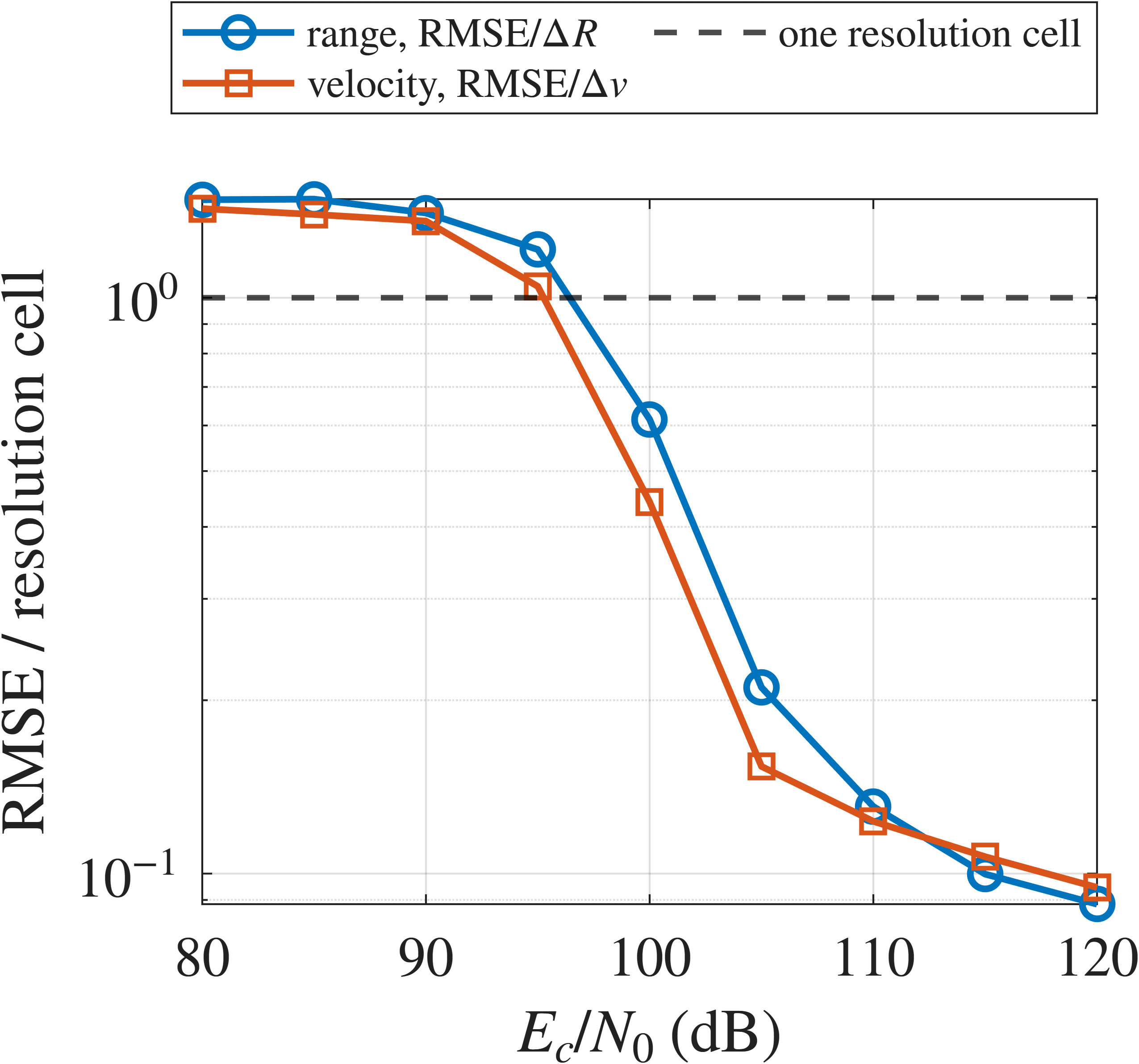}%
        \label{fig_nrmse}}
    \hfil
    \subfloat[]{\includegraphics[width=0.49\columnwidth]{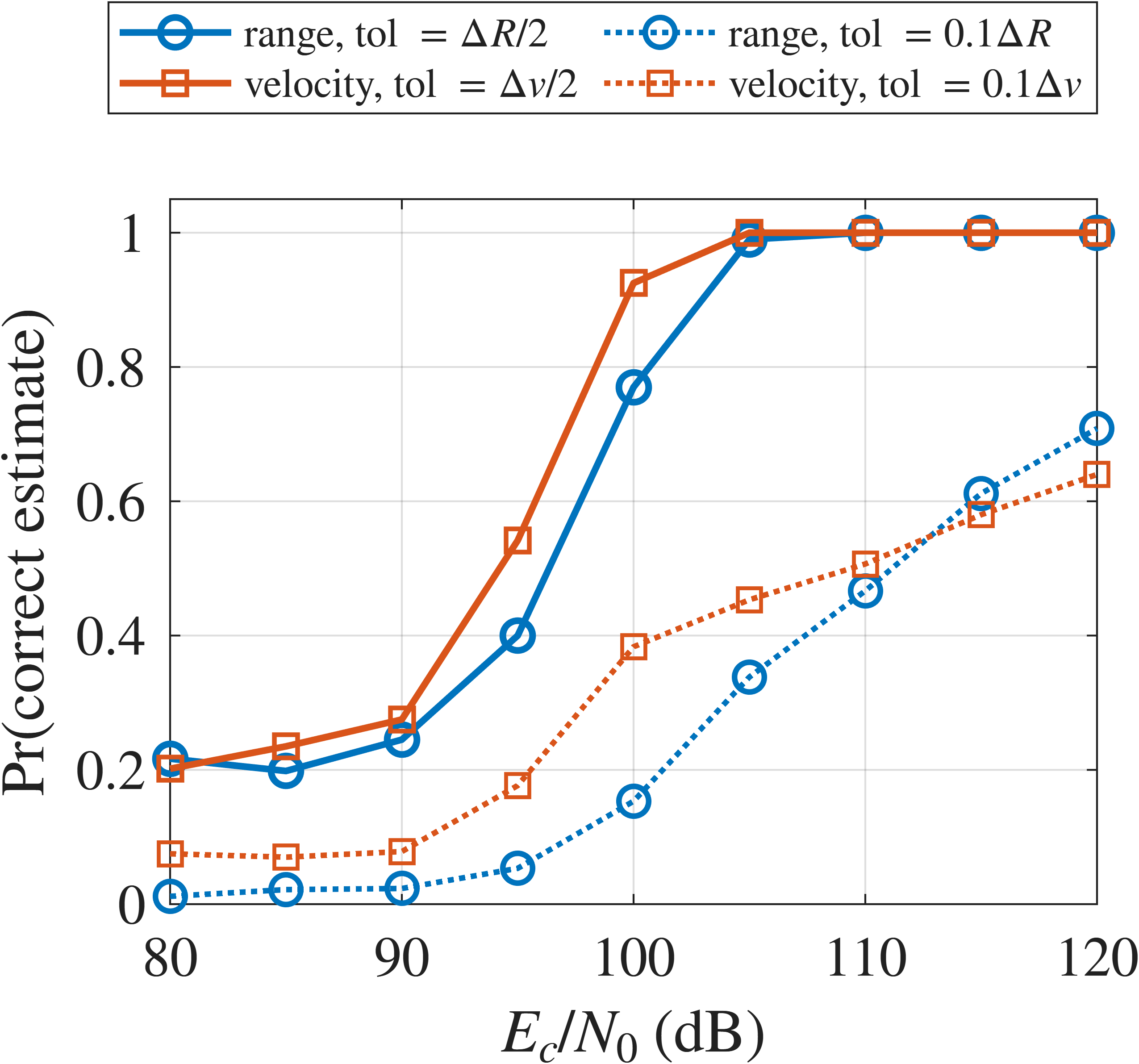}%
        \label{fig_pcorr}}
    \caption{Estimation accuracy vs. $E_c/N_0$: (a) resolution-normalized RMSE \eqref{eq_nrmse}, sub-cell below the dashed line; (b) probability of a correct estimate \eqref{eq_pcorr_R} and \eqref{eq_pcorr_v} at tolerances $\varepsilon=1/2$ (solid) and $0.1$ (dotted).}
    \label{fig_estimation}
\end{figure}

Figs.~\ref{fig_nrmse} and~\ref{fig_pcorr} report the resolution-aware estimation metrics of Section~\ref{subsubsec_estimation_metrics}. Beyond a moderate $E_c/N_0$, $\mathrm{NRMSE}_R$ and $\mathrm{NRMSE}_v$ fell below one cell (sub-resolution accuracy), and the probability of a correct estimate reached one at the half-cell tolerance, while the tight tolerance $\varepsilon=0.1$ required a higher $E_c/N_0$. These results confirm that sub-bin interpolation in Eq.~\eqref{eq_rd_estimates} attains near-exact range and velocity estimates at high SNR.

\section{Conclusion}
\label{sec_conclusion}

In this study, we analyzed a communication-optimized RIS-aided ISAC system~\cite{Triwidyastuti2026arxiv} in a cluttered environment. For the UE link over Nakagami-$m$ fading, we derived an exact no-RIS ergodic-capacity expression, a moment-matched Gamma approximation for the RIS-assisted link, and a Jensen upper bound. For sensing, we derived average return powers and a closed-form average SCNR at the direct range-Doppler cell. The results revealed that all static clutter, both RIS-induced (scaling with $L$) and environmental (determined by geometry), entered through a common zero-Doppler leakage mechanism. Monte Carlo simulation results confirmed the analysis.

\bibliographystyle{IEEEtran}
\bibliography{references}

@misc{Triwidyastuti2026arxiv,
	title={{Communication-Centric RIS-Assisted ISAC: Signal Modeling and BER Analysis}}, 
	author={Yosefine Triwidyastuti and Tri Nhu Do},
	year={2026},
	eprint={2606.28924},
	archivePrefix={arXiv},
	primaryClass={eess.SP},
	url={https://arxiv.org/abs/2606.28924}, 
}

@ARTICLE{VaeziCOMST2026,
	author={Vaezi, Mojtaba and Baduge, Gayan Amarasuriya Aruma and Ollila, Esa and Vorobyov, Sergiy A.},
	journal={IEEE Commun. Surveys Tuts.},
	title={{A Tutorial on AI-Empowered Integrated Sensing and Communications}}, 
	year={2026},
	volume={28},
	number={},
	pages={4980-5013},
	doi={10.1109/COMST.2026.3665143},
	ISSN={1553-877X},
	month={},}

@BOOK{Gradshteyn2007,
	title     = "Table of Integrals, Series, and Products",
	author    = "Gradshteyn, I. S. and Ryzhik, I. M.",
	publisher = "Academic Press",
	edition   = 7,
	year      = 2007,
	address   = "Burlington, MA"
}

@BOOK{Goldsmith2005,
	title     = "Wireless Communications",
	author    = "Goldsmith, Andrea",
	publisher = "Cambridge University Press",
	month     =  aug,
	year      =  2005,
	address   = "Cambridge"
}

@BOOK{Tse2005,
	title     = "Fundamentals of Wireless Communication",
	author    = "Tse, David and Viswanath, Pramod",
	publisher = "Cambridge University Press",
	year      =  2005,
	address   = "Cambridge"
}

@BOOK{Barkat2005,
  title     = "Signal Detection and Estimation",
  author    = "Barkat, Mourad",
  publisher = "Artech House",
  series    = "Radar Library",
  edition   =  2,
  month     =  aug,
  year      =  2005,
  address   = "Norwood, MA",
}

@BOOK{Richards2014,
  title     = "Fundamentals of radar signal processing, second edition",
  author    = "Richards, Mark A",
  publisher = "McGraw-Hill Professional",
  edition   =  2,
  month     =  jan,
  year      =  2014,
  address   = "New York, NY",
  language  = "en"
}

@article{BouzabiaTVT2025,
  title = {{Generative AI-Empowered Resilient Adaptive ISAC Against Adversarial Machine Learning Attacks}},
  volume = {74},
  ISSN = {1939-9359},
  DOI = {10.1109/tvt.2025.3583924},
  number = {12},
  journal = {IEEE Trans. Veh. Technol.},
  publisher = {Institute of Electrical and Electronics Engineers (IEEE)},
  author = {Bouzabia,  Hamda and Do,  Tri Nhu and Kaddoum,  Georges},
  pages = {18984-19000},
  month = Dec,
  year = {2025}
}

@ARTICLE{DoTCOM2021,
	author={Do, Tri Nhu and Kaddoum, Georges and Nguyen, Thanh Luan and da Costa, Daniel Benevides and Haas, Zygmunt J.},
	journal={IEEE Trans. Commun.}, 
	title={{Multi-RIS-Aided Wireless Systems: Statistical Characterization and Performance Analysis}}, 
	year={2021},
	volume={69},
	number={12},
	pages={8641-8658},
	doi={10.1109/TCOMM.2021.3117599},
	ISSN={1558-0857},
	month={Dec},}
\end{document}